\pdfoutput=1
\documentclass[letterpaper,abstract=true]{scrartcl}
\usepackage[letterpaper,margin=2.5cm,includefoot]{geometry}

\usepackage[utf8]{inputenc}
\usepackage[english]{babel}
\usepackage[T1]{fontenc}

\usepackage{lmodern}
\usepackage{amsmath,amssymb,mathtools,bm}
\usepackage{csquotes}
\usepackage{amstext}
\usepackage{amsthm}
\usepackage{bbm}
\usepackage{enumerate}
\usepackage{aliascnt}
\usepackage{hyperref}
\usepackage[nameinlink,capitalise]{cleveref}
\usepackage{dsfont}
\usepackage{comment}
\usepackage{color}

\usepackage{fancyvrb}

\makeatletter

\numberwithin{equation}{section}

\newtheorem{thm}{Theorem}[section]

\newaliascnt{lem}{thm}
\newtheorem{lem}[lem]{Lemma}
\aliascntresetthe{lem}

\newaliascnt{prop}{thm}
\newtheorem{prop}[prop]{Proposition}
\aliascntresetthe{prop}

\newaliascnt{cor}{thm}
\newtheorem{cor}[cor]{Corollary}
\aliascntresetthe{cor}

\theoremstyle{definition}

\newaliascnt{defn}{thm}

\aliascntresetthe{defn}

\newtheorem{hyp}{Hypothesis}
\renewcommand*{\thehyp}{\Alph{hyp}}

\theoremstyle{remark}

\newaliascnt{rem}{thm}
\newtheorem{rem}[rem]{Remark}
\aliascntresetthe{rem}

\newaliascnt{ex}{thm}
\newtheorem{ex}[ex]{Example}
\aliascntresetthe{ex}

\crefname{thm}{Theorem}{Theorems}
\Crefname{thm}{Theorem}{Theorems}

\crefname{lem}{Lemma}{Lemmas}
\Crefname{lem}{Lemma}{Lemmas}

\crefname{prop}{Proposition}{Propositions}
\Crefname{prop}{Proposition}{Propositions}

\crefname{cor}{Corollary}{Corollaries}
\Crefname{cor}{Corollary}{Corollaries}

\crefname{defn}{Definition}{Definitions}
\Crefname{defn}{Definition}{Definitions}

\crefname{rem}{Remark}{Remarks}
\Crefname{rem}{Remark}{Remarks}

\crefname{ex}{Example}{Examples}
\Crefname{ex}{Example}{Examples}

\crefname{hyp}{Hypothesis}{Hypotheses}
\Crefname{hyp}{Hypothesis}{Hypotheses}
\makeatletter
\renewcommand{\@upn}{} \makeatother

\usepackage[inline]{enumitem}
\newlist{enumthm}{enumerate}{1} \setlist[enumthm]{label=\upshape(\roman*),ref=\thethm~(\roman*)}  \crefalias{enumthmi}{thm} \newlist{enumcor}{enumerate}{1}
\setlist[enumcor]{label=\upshape(\roman*),ref=\thecor~(\roman*)}
\crefalias{enumcori}{cor}
\newlist{enumlem}{enumerate}{1}
\setlist[enumlem]{label=\upshape(\roman*),ref=\thelem~(\roman*)}
\crefalias{enumlemi}{lem}
\newlist{enumprop}{enumerate}{1}
\setlist[enumprop]{label=\upshape(\roman*),ref=\theprop~(\roman*)}
\crefalias{enumpropi}{prop}
\newlist{enumhyp}{enumerate}{1}
\setlist[enumhyp]{label=\upshape(\roman*),ref=\thehyp~(\roman*)}
\crefalias{enumhypi}{hyp}
\newlist{enumproof}{enumerate*}{1}
\setlist[enumproof]{label=\upshape(\roman*)}
\newlist{enumdef}{enumerate}{1}
\setlist[enumdef]{label=\upshape(\roman*),ref=\thedefn~(\roman*)}
\crefalias{enumdefi}{defn}
\newlist{enumrem}{enumerate}{1}
\setlist[enumrem]{label=\itshape(\alph*),ref=\therem~(\alph*)}
\crefalias{enumrem}{rem}

\makeatletter
\newcounter{subcreftmpcnt} \newcommand\romansubformat[1]{(\roman{#1})} \def\subcref{\@ifstar\@@subcref\@subcref}
\newcommand\@subcref[2][\romansubformat]{\ifcsname r@#2@cref\endcsname
	\cref@getcounter {#2}{\mylabel}\setcounter{subcreftmpcnt}{\mylabel}\hyperref[#2]{\romansubformat{subcreftmpcnt}}\else ?? \fi}   
\newcommand\@@subcref[2][\romansubformat]{\ifcsname r@#2@cref\endcsname
	\cref@getcounter {#2}{\mylabel}\setcounter{subcreftmpcnt}{\mylabel}\romansubformat{subcreftmpcnt}\else ?? \fi}   
\makeatother

\makeatletter
\DeclareRobustCommand{\crefnosort}[1]{\begingroup\@cref@sortfalse\cref{#1}\endgroup
}
\makeatother

\def\endstepsymbol{$\lozenge$}
\def\endclaimsymbol{$\lozenge$}
\newcounter{proofstep}
\AtBeginEnvironment{proof}{\setcounter{proofstep}{0}}

\crefname{proofstep}{Step}{Steps}
\Crefname{proofstep}{Step}{Steps}
\newcounter{proofclaim}
\AtBeginEnvironment{proof}{\setcounter{proofclaim}{0}}

\crefname{proofclaim}{Claim}{Claims}
\Crefname{proofclaim}{Claim}{Claims}

\newcommand{\cB}{{\mathcal B}}

\newcommand{\cH}{{\mathcal H}}

\newcommand{\fh}{{\mathfrak h}}

\newcommand{\BC}{{\mathbb C}}

\newcommand{\BN}{{\mathbb N}}
\newcommand{\BR}{{\mathbb R}}

\usepackage{dsfont}

\newcommand{\dsone}{{\mathds 1}}

\usepackage{mathrsfs}

\newcommand{\sfd}{{\mathsf d}}
\newcommand{\sfi}{{\mathsf i}}

\newcommand{\IN}{\BN}\newcommand{\IR}{\BR}\newcommand{\IC}{\BC}

\newcommand{\NN}{\BN}

\newcommand{\eps}{\varepsilon}

\renewcommand{\i}{\sfi}\newcommand{\Id}{\dsone} \renewcommand{\d}{\sfd}

\renewcommand{\Re}{\operatorname{Re}}

\newcommand{\supp}{\operatorname{supp}}

\DeclareFontFamily{U}{mathx}{\hyphenchar\font45}
\DeclareFontShape{U}{mathx}{m}{n}{
	<5> <6> <7> <8> <9> <10>
	<10.95> <12> <14.4> <17.28> <20.74> <24.88>
	mathx10
}{}
\DeclareSymbolFont{mathx}{U}{mathx}{m}{n}
\DeclareFontSubstitution{U}{mathx}{m}{n}
\DeclareMathAccent{\widecheck}{0}{mathx}{"71}
\DeclareMathAccent{\wideparen}{0}{mathx}{"75}

\DeclareFontFamily{OMX}{MnSymbolE}{}
\DeclareFontShape{OMX}{MnSymbolE}{m}{n}{
	<-6>  MnSymbolE5
	<6-7>  MnSymbolE6
	<7-8>  MnSymbolE7
	<8-9>  MnSymbolE8
	<9-10> MnSymbolE9
	<10-12> MnSymbolE10
	<12->   MnSymbolE12}{}
\DeclareSymbolFont{mnlargesymbols}{OMX}{MnSymbolE}{m}{n}
\SetSymbolFont{mnlargesymbols}{bold}{OMX}{MnSymbolE}{b}{n}
\DeclareMathDelimiter{\llangle}{\mathopen}{mnlargesymbols}{'164}{mnlargesymbols}{'164}
\DeclareMathDelimiter{\rrangle}{\mathclose}{mnlargesymbols}{'171}{mnlargesymbols}{'171}
\DeclareMathDelimiter{\lsem}{\mathopen}{mnlargesymbols}{'102}{mnlargesymbols}{'102}
\DeclareMathDelimiter{\rsem}{\mathclose}{mnlargesymbols}{'107}{mnlargesymbols}{'107}
\DeclareMathDelimiter{\langlebar}{\mathopen}{mnlargesymbols}{'152}{mnlargesymbols}{'152}
\DeclareMathDelimiter{\ranglebar}{\mathclose}{mnlargesymbols}{'157}{mnlargesymbols}{'157}
\DeclareMathDelimiter{\lWavy}{\mathopen}{mnlargesymbols}{'137}{mnlargesymbols}{'137}
\DeclareMathDelimiter{\rWavy}{\mathopen}{mnlargesymbols}{'137}{mnlargesymbols}{'137}

\newcommand{\chr}{\mathbf 1}
\newcommand{\abs}[1]{\lvert#1\rvert}\newcommand{\Abs}[1]{\left\lvert#1\right\rvert}
\newcommand{\norm}[1]{\lVert#1\rVert}
\newcommand{\Norm}[1]{\left\lVert#1\right\rVert}

\newcommand{\FGamma}{\Gamma}
\newcommand{\dG}{\sfd\FGamma}

\newcommand{\nn}[1]{\Norm{#1}}

\renewcommand{\:}{\colon}

\usepackage{xcolor}\usepackage{cancel}
\definecolor{green}{rgb}{0.0, 0.5, 0.5}
\definecolor{yellow}{rgb}{0.5, 0.5, 0}
\definecolor{lgray}{gray}{0.9}
\definecolor{llgray}{gray}{0.95}
\definecolor{lllgray}{gray}{0.975}

\definecolor{darkcerulean}{rgb}{0.03, 0.27, 0.49} 	\definecolor{darkcoral}{rgb}{0.8, 0.36, 0.27}

\newcommand{\Af}{\mathfrak{A}}

\newcommand{\Nf}{\mathcal{N}}

\newcommand{\po}{\mathsf{p}}
\newcommand{\ho}{h}

\usepackage[backend=biber,style=alphabetic,giveninits=true,maxnames=5,maxalphanames=4, ,url=true,isbn=false,eprint=false,alldates=year]{biblatex}
\DeclareFieldFormat{doilink}{\iffieldundef{doi}{#1}{\href{http://dx.doi.org/\thefield{doi}}{#1}}}

\DeclareBibliographyDriver{article}{\usebibmacro{bibindex}\usebibmacro{begentry}\usebibmacro{author/translator+others}\setunit{\labelnamepunct}\newblock
	\usebibmacro{title}\newunit
	\printlist{language}\newunit\newblock
	\usebibmacro{byauthor}\newunit\newblock
	\usebibmacro{bytranslator+others}\newunit\newblock
	\printfield{version}\newunit\newblock
	\usebibmacro{in:}\printtext[doilink]{\usebibmacro{journal+issuetitle}\newunit
		\usebibmacro{byeditor+others}\newunit
		\usebibmacro{note+pages}}\newunit\newblock
	\iftoggle{bbx:isbn}
	{\printfield{issn}}
	{}\newunit\newblock
	\usebibmacro{doi+eprint+url}\newunit\newblock
	\usebibmacro{addendum+pubstate}\setunit{\bibpagerefpunct}\newblock
	\usebibmacro{pageref}\usebibmacro{finentry}}

\newbibmacro{string+doi}[1]{\iffieldundef{doi}{#1}{\href{https://doi.org/\thefield{doi}}{#1}}}
\DeclareFieldFormat{title}{\usebibmacro{string+doi}{\mkbibemph{#1}}}
\DeclareFieldFormat[inproceedings]{title}{\usebibmacro{string+doi}{\mkbibquote{#1}}}
\DeclareFieldFormat[incollection]{title}{\usebibmacro{string+doi}{\mkbibquote{#1}}}
\DeclareFieldFormat[unpublished]{title}{\usebibmacro{string+doi}{\mkbibquote{#1}}}

\renewbibmacro{in:}{\ifentrytype{article}{}{\printtext{\bibstring{in}\intitlepunct}}}

\newcommand{\Fc}{\mathcal{F}}

\newcommand{\cs}[1]{\left< #1 \right>}

\newcommand{\Def}{\mathcal{D}}

\newcommand{\Hb}{\mathbf{H}}
\newcommand{\Ec}{\mathcal{E}}

\begin{document}

\title{Discontinuity of Continuum Fermion Dynamics}
\author{Oliver Siebert\footnote{ \href{mailto:osiebert@ucdavis.edu}{osiebert@ucdavis.edu}} \\ \small{Department of Mathematics, University of California, Davis}}

\maketitle

\begin{abstract}
	We study the Heisenberg dynamics of continuum fermions in $\IR^d$ interacting through a pair potential. For a large class of nonconstant pair potentials and under suitable assumptions on the dispersion relation and external potential, we prove that the dynamics fails to be pointwise-norm continuous even on the gauge-invariant CAR algebra.  An automatic-continuity argument then shows that neither the CAR algebra nor its gauge-invariant subalgebra is invariant under the dynamics, with non-invariance occurring for almost all times.
\end{abstract}

\section{Introduction}

In the algebraic approach to quantum systems in infinite volume, a
$C^*$-dynamical system
$
\bigl(\Af,(\tau_t)_{t\in\IR}\bigr)
$
consists of a $C^*$-algebra of observables $\Af$ and a one-parameter
group of $*$-automorphisms $(\tau_t)_{t\in\IR}$ such that
\begin{align}
\label{eq:norm cont}	
\lim_{t\to0}\norm{\tau_t(A)-A}=0,
\qquad \text{ for all }A\in\Af.
\end{align}

This continuity property is also commonly referred to as strong continuity of the dynamics, but throughout this paper, we call it
pointwise-norm continuity in order to avoid confusion with  continuity in
the strong operator topology.  A $C^*$-dynamical system provides a
state-independent formulation of the time evolution and is the
natural starting point for the study of ground states, equilibrium
states satisfying the KMS condition, and nonequilibrium phenomena in infinite volume within 
algebraic quantum statistical mechanics
\cite{BR1,BR2,sakai1991operator}. The observable algebra should be large enough to contain the physically relevant observables and to be invariant under the Heisenberg evolution, but sufficiently restrictive for this evolution to be pointwise-norm continuous and so that infinite-volume equilibrium states can be canonically defined.

For fermionic systems, the free dynamics has this property on the CAR
algebra: if the many-body Hamiltonian is the second
quantization of a self-adjoint one-particle operator $h$, then the
Heisenberg evolution is the Bogoliubov dynamics induced by the
one-particle group $e^{\i th}$.  Since the fermionic creation and
annihilation operators are bounded and satisfy
$
\norm{a^*(f)}=\norm{a(f)}=\norm f,
$
the strong continuity of $e^{\i th}$ on the one-particle space
immediately implies pointwise-norm continuity on the entire CAR algebra.

However, the free case is a peculiarity, as interacting fermions in the continuum in general face a
many-body ultraviolet problem. Under suitable assumptions on the potentials, the Hamiltonian
is self-adjoint on every fixed particle sector, hence also on the full Fock space, and the
resulting dynamics is continuous in the strong
operator topology. Nevertheless, this does not imply pointwise-norm continuity or invariance of the CAR algebra. The operator norm simultaneously probes all particle sectors, and the interaction energy seen by one particle can grow indefinitely as the number of surrounding particles increases.

This obstruction was already discussed in
\cite[Introduction to Sec.~6.3, p.~354]{BR2}.  Although
the Pauli principle prevents several fermions from occupying the same quantum state,  arbitrarily many modes with
increasing momenta can be packed into small regions.  A particle at a fixed distance from such a high-density cloud then experiences an interaction whose strength grows proportionally to the number of particles.  It was therefore expected that the usual
point-particle interaction should not generate a pointwise-norm continuous dynamics on the CAR algebra.  
This expectation was repeated in
\cite{GebertNachtergaeleReschkeSims.2020}, where it was observed that
an unregularized interaction has an unbounded commutator with a
creation or annihilation operator.  However, such a commutator argument by
itself proves neither norm discontinuity nor non-invariance of the CAR algebra.   A different heuristic no-go argument against the preservation of quasi-local observables was given in \cite{radin}, based on an exactly solvable one-dimensional hard-core model of distinguishable particles.

The ultraviolet problem was addressed in earlier work through the introduction of UV regularizations in the interaction.  Narnhofer and Thirring \cite{narnhofer1990quantum,narnhofer1991galilei}  introduced a cutoff at high relative momenta and proved that it generates a $C^*$-dynamical system over the CAR algebra. The idea is that the cutoff causes the  Dyson series to be summable in operator norm. 
Their construction
also gave equilibrium states, whose uniqueness in a high-temperature,
low-density regime was later proved in \cite{jakel1995uniqueness}.

Another UV regularization was introduced in
\cite{GebertNachtergaeleReschkeSims.2020}, where the creation and
annihilation operators in the interaction are smeared with a fixed
Gaussian profile. Restricting the interaction to finite volume then yields a bounded perturbation of the free, well-behaved dynamics. 
 They proved Lieb--Robinson bounds which show the convergence of the dynamics, uniformly on compact time intervals, as the volume tends to the whole space. Later in
 \cite{hinrichs2024lieb}, the
 assumptions on the one-particle and pair potentials were
 relaxed, and the associated propagation bounds improved.  
 Another construction based on
 lattice-localized frames, suitable
 for continuum fermions in magnetic fields, was developed in
 \cite{bachmann2024liebrobinson}, which allowed for Lieb--Robinson proof techniques almost like for discrete systems. In fact, for lattice systems, Lieb--Robinson bounds have become a standard tool in order to obtain a well-behaved infinite-volume dynamics \cite{td1,td2,td3,td4,irreversible,deuchert2025dynamics}.

 The bosonic situation is even more troublesome.  The standard bounded formulation
 of the canonical commutation relations is the Weyl algebra generated by exponentials of the field operators, on which
 even the free dynamics is not pointwise-norm continuous. In \cite{dubin} it was also shown that the norm convergence of finite-volume dynamics fails for the ideal Bose gas and the (fermionic, mean-field) BCS model.
 Moreover, interacting dynamics may fail to preserve the algebra already for finitely many degrees of freedom.  In fact,
 in \cite{FannesVerbeure1974} it was proven that a non-relativistic Schrödinger operator with potential $V\in L^\infty \cap L^1$ does not leave the Weyl algebra invariant unless $V=0$. The resolvent algebra of Buchholz and Grundling \cite{bh}, generated by
 resolvents of the fields, resolves this invariance issue.  For nonrelativistic
 Bose fields, Buchholz showed that a slight extension of the resolvent algebra is also invariant under many-body dynamics with unregularized pair interactions \cite{bh1,bh2}.  The dynamics is only continuous on each fixed
 particle sector, but time averaging yields a pointwise-norm continuous
 $C^*$-dynamical subsystem.
 
A similar extension approach as in \cite{bh1,bh2} was taken in \cite{my_tdlimit} for fermions. There, a slightly extended CAR algebra was constructed by completing the CAR algebra with respect to the family of operator norms
on fixed particle sectors.  The necessity of
the extension was motivated by the high-density argument described
above but remained without a rigorous discontinuity and non-invariance theorem, which seemed to be missing in the current literature.

The purpose of this note is to provide this missing negative
result.  Indeed, we will prove such a negative result for a large class of Hamiltonians whose $n$-particle sectors are of the form 
$$
H_n = \sum_{j=1}^n\bigl(\omega(\po_j)+U(x_j)\bigr)
+ \sum_{\substack{1\leq i,j\leq n\\i\neq j}} W(x_i-x_j),
$$
with general assumptions on the dispersion relations $\omega$, the pair potential $W$ and the external potential $U$. This includes the non-relativistic dispersion relation
$\omega(p)=p^2$ in dimensions $d\geq2$ and the relativistic
dispersion relation
$\omega(p)=\sqrt{p^2+1}$ in every dimension $d\geq1$.  The pair
potential may be unbounded and singular, including the Coulomb interaction in $d\geq3$. The external potential includes smooth bounded fields and multi-centre Coulomb potentials  for $\omega(p) = p^2$ in $d=3$.

 We will first show the violation of the pointwise-norm continuity \eqref{eq:norm cont}, which occurs already for a single creation or annihilation operator, or for a single pair of creation and annihilation operators in the gauge-invariant case. Using the separability of the CAR algebra over $L^2$, an automatic-continuity argument will then imply non-invariance for almost every time. This result provides a justification for the extension of the CAR algebra as in \cite{my_tdlimit}. It also shows that the
 usual application of Lieb--Robinson bounds for the construction of
 pointwise-norm continuous infinite-volume dynamics on the CAR
 algebra as in \cite{GebertNachtergaeleReschkeSims.2020,hinrichs2024lieb}
 cannot extend to the unregularized interactions considered here.

In \Cref{sec:setup-main} we introduce the Fock space formalism, formulate
the hypotheses on the dispersion relation and the potentials, and
state the main results. Subsequently, we discuss some specific examples where these hypotheses are satisfied. \Cref{sec:proof} contains all the proofs, beginning with a brief outline. In particular, the proof of the main result, \Cref{thm:main-discontinuity}, can be found in \Cref{sec:main result proof}.

\section{Setup and Main Results}
\label{sec:setup-main}

\subsection{Fermionic Fock space and the CAR algebra}

Let $d\geq 1$, $\fh=L^2(\IR^d)$ be the one-particle space, and 
\[
    \Fc:=\bigoplus_{n=0}^{\infty}\bigwedge\nolimits^n\fh,
    \qquad
    \bigwedge\nolimits^0\fh=\IC
\]
be the standard fermionic Fock space over $\fh$. For $f\in\fh$,
let $a^*(f)$ and $a(f)$ denote the usual fermionic creation and
annihilation operators on $\Fc$.  They are bounded with $\nn{a^*(f)} = \nn{a(f)} = \nn{f}$, and satisfy the CAR
\[
    \{a(f),a^*(g)\}=\langle f,g\rangle\Id,
    \qquad
    \{a(f),a(g)\}=\{a^*(f),a^*(g)\}=0.
\]

For a self-adjoint operator $T$ on $\fh$, the lift to the Fock space $\Fc$ is the self-adjoint operator
\[
\dG(T):= 0 \oplus \bigoplus_{n=1}^\infty \overline{
	\left.
	\sum_{j=1}^n
	\Id^{\otimes(j-1)}\otimes T
	\otimes\Id^{\otimes(n-j)}
	\right|_{\bigwedge_{\mathrm{alg}}^n\Def(T)}
},
\]
where the overline denotes the closure and the direct sum is taken on its natural maximal domain
\[
\Def(\dG(T))
=
\left\{
(\Psi_n)_{n\in \IN_0}\in\Fc:
\Psi_n\in\Def(\dG(T)_n),\
\sum_{n=0}^\infty\norm{\dG(T)_n\Psi_n}^2<\infty
\right\},
\]
with $\dG(T)_n$ denoting  its restriction to the $n$-particle sector. 
The global number operator, and for a measurable set $A\subseteq\IR^d$, the local number operator are defined as
\[
\Nf:=\dG(\Id),
\qquad
\Nf_A:=\dG(\chr_A),
\]
where $\chr_A$ denotes multiplication by its indicator function.
By abuse of notation, when these operators act on an
$n$-particle vector, we suppress the sector subscript:
$\dG(T)$, $\Nf$, and $\Nf_A$ then mean
$\dG(T)_n$, $\Nf_n$, and $\Nf_{A,n}$, respectively.

The CAR algebra is the $C^*$-algebra generated by the creation and annihilation operators,
\[
    \Af_{\mathrm{CAR}}
    :=C^*\bigl(a(f) :f\in\fh\bigr)
    \subset\cB(\Fc).
\]
The gauge-invariant CAR algebra is the $C^*$-subalgebra of $\Af_{\mathrm{CAR}}$, which does not change the number of particles, and can be defined via
\[
    \Af^0_{\mathrm{CAR}}
    :=\{A\in\Af_{\mathrm{CAR}}:
          e^{\i\vartheta \Nf}Ae^{-\i\vartheta \Nf} =A\text{ for every }\vartheta\in\IR\}.
\]
It is generated by all observables with the same number of creation and annihilation operators. In particular, 
\[
    n_f:=a^*(f)a(f) \in \Af^0_{\mathrm{CAR}}, \qquad  f\in\fh.
\]
For normalized $f$, this is the projection onto occupation
of the one-particle mode $f$.

\subsection{Definition of the Hamiltonian and assumptions}

We first collect all the necessary hypotheses on the dispersion relation of the kinetic energy, on the interaction potential and on the external potential. Then we define the Hamiltonian. 

In the following let  $\po=-\i\nabla$ denote the vector of momentum operators in $\IR^d$ and write $\cs{x} := \sqrt{\abs x^2+1}$ for $x \in \IR^d$. 
\begin{hyp}[Dispersion relation]
\label{hyp:dispersion}
The function $\omega\:\IR^d\to[0,\infty)$ is measurable and even.
There are constants $C_\omega>0$ and
$0<\nu\leq\min\{2,d\}$ such that
\[
    \omega(p)\leq C_\omega\langle p\rangle^\nu,
    \qquad p\in\IR^d.
\]
Moreover, for every real-valued $\chi\in C_c^\infty(\IR^d)$,
the commutator initially defined on $C_c^\infty(\IR^d)$ extends to
the form domain of $\omega(\po)$ and satisfies
\begin{align}
    \norm{[\omega(\po),\chi]u}^2
    \leq
    C_\chi\bigl(
        \norm{\omega(\po)^{1/2} u}^2+\norm u^2
    \bigr),
    \qquad u\in\Def(\omega(\po)^{1/2}).
    \label{eq:dispersion est}
\end{align}
\end{hyp}

\begin{ex}[Examples of dispersion relations]
\label{rem:dispersion-examples}
\Cref{hyp:dispersion} holds for
$
\omega(p)=\langle p\rangle^\nu
$
whenever $0<\nu\leq\min\{2,d\}$, by standard
pseudodifferential commutator estimates, see e.g. the symbolic composition and commutator calculus for $S^m_{1,0}$ symbols in \cite{hormander}. Thus, the
relativistic dispersion $\omega(p)=\sqrt{\abs{p}^2+1}$ is covered in
every dimension $d\geq1$.  It also holds for the standard non-relativistic dispersion relation
$\omega(p)=p^2$ when $d\geq2$. 
\end{ex}

Let $W\:\IR^d\to\IR$ be the pair potential.  We use the following
basic hypothesis.

\begin{hyp}[Pair potential]
\label{hyp:pair-potential}
The potential $W$ is measurable, real-valued, and even, and it admits
a decomposition
\[
    W=W_2+W_\infty,
    \qquad
    W_2\in L^2(\IR^d),
    \qquad
    W_\infty\in L^\infty(\IR^d).
\]
Moreover, we assume $W$ to be bounded from below, i.e., there is a constant $C_W\geq0$ such that
\[
    W(x)\geq-C_W
    \qquad\text{for almost every }x\in\IR^d.
\]
\end{hyp}

\begin{ex}[Repulsive Riesz and Coulomb pair interactions]
	\label{cor:riesz-coulomb-main}
	Apart from bounded potentials, the hypothesis also allows for potentials with singularities.
	Let $d\geq1$ and
	$\kappa>0$, and consider
	\[
	W_\sigma(x) =\frac{\kappa}{\abs{x}^\sigma},
	\qquad 0<\sigma<\frac d2.
	\]
	Then $W_\sigma$ satisfies
	\cref{hyp:pair-potential}, since
	\[
	W_\sigma
	=
	W_\sigma\chr_{\{\abs x<1\}}
	+
	W_\sigma\chr_{\{\abs x\geq1\}}
	\]
	where the first term is in $L^2(\IR^d)$ precisely when $2\sigma<d$, while
	the second term is bounded.  
	In particular, the repulsive Coulomb potential
	$
	W_{1}(x)=\frac{\kappa}{\abs x}
	$
	in any $d\geq 3$ is covered.
\end{ex}

We also allow a real-valued external potential $U\:\IR^d\to\IR$.
Its negative part is assumed to be form-small with respect to
$\omega(\po)$.

\begin{hyp}[External potential]
\label{hyp:external-form}
One has $U\in L^1_{\mathrm{loc}}(\IR^d)$, and there are constants
$0\leq a<1$ and $b\geq0$ such that
\[
    \int_{\IR^d}U_-(x)\abs{u(x)}^2\,\d x
    \leq a \nn{\omega(\po)^{1/2}u}^2 +b\norm{u}^2,
    \qquad u\in C_c^\infty(\IR^d).
\]
\end{hyp}
Under \cref{hyp:dispersion,hyp:external-form}, the one-particle Hamiltonian
\[
\ho=\omega(\po)+U
\]
is, by the KLMN theorem, defined as a self-adjoint, lower-bounded form sum. 

 Additionally, for the proof one needs local regularity of the one-particle propagation near the two test modes. 
\begin{hyp}[Local propagation]
	\label{hyp:local-propagation}
	There are an open set $\Omega_U\subseteq\IR^d$ and a number
	$m>d/2$ such that, for every $u\in C_c^\infty(\Omega_U)$,
	\[
	\norm{e^{\i t\ho}u-u}_{H^m}\leq C_u\abs t,
	\qquad \abs t\leq1.
	\]
\end{hyp}
\begin{ex}[Smooth and Coulomb external fields]
	\label{cor:external-examples}
	\Cref{hyp:external-form,hyp:local-propagation} hold in each of
	the following cases.
	\begin{enumerate}[label=\upshape(\arabic*)]
		\item \label{it:external simple} \Cref{hyp:dispersion} holds and 
		$U\in W^{k,\infty}(\IR^d; \IR)$, where
		$
		k=\left\lfloor\frac d2\right\rfloor+1,
		$ with $\Omega_U = \IR^d$. In particular, they hold in the translation-invariant case $U \equiv 0$. 
		\item \label{it:external coulomb}  $d=3$, $\omega(p)=p^2$, and
		\[
		U(x)=U_0(x)+\sum_{\ell=1}^M
		\frac{\lambda_\ell}{\abs{x-R_\ell}},
		\qquad
		U_0\in W^{2,\infty}(\IR^3; \IR),
		\quad \lambda_\ell\in\IR,
		\]
			with
		$\Omega_U=\IR^3\setminus\{R_1,\ldots,R_M\}$.
	\end{enumerate}
\end{ex}
 On the $n$-particle space $\bigwedge^n\fh$, we define $H_n$ by the closed, lower-bounded quadratic form
\[
    \Ec_n[\Psi]
    = \nn{\dG(\omega(\po))^{1/2}\Psi }^2
       +\sum_{j=1}^n\int_{\IR^{dn}}U(x_j)\abs{\Psi(X)}^2\,\d X
      +\sum_{\substack{1\leq i,j\leq n\\i\neq j}}
       \int_{\IR^{dn}}W(x_i-x_j)\abs{\Psi(X)}^2\,\d X.
\]
The assumptions above
ensure that the form domain is dense and that $H_n$ is
self-adjoint and bounded from below. Note that the form domain $\Def(\Ec_n)$ is given as the intersection of the kinetic form domain $\Def(\dG(\omega(\po))^{1/2})$ with the multiplication-form domains of the positive parts of the potentials $U$ and $W$.  We then set
\[
    \Hb:=\bigoplus_{n=0}^{\infty}H_n.
\]
The direct sum is self-adjoint on its natural domain and hence defines
a strongly continuous unitary group on $\Fc$.  Its Heisenberg
dynamics on $\cB(\Fc)$ is
\[
    \tau_t(A):=e^{\i t\Hb}Ae^{-\i t\Hb},
    \qquad t\in\IR.
\]

\subsection{Main results}

Our main result is norm discontinuity for a fixed,
particle-number-preserving CAR observable. Thus, norm discontinuity already occurs in the
gauge-invariant CAR algebra.  The relevant time scale is
the inverse of the number of cluster particles used to test the operator
norm.
\begin{thm}[Discontinuity]
\label{thm:main-discontinuity}
Suppose that $\omega$ satisfies \cref{hyp:dispersion}. Let $W$ satisfy
\cref{hyp:pair-potential} and assume that it has a representative with two distinct nonzero
continuity points $z_1,z_2\in\IR^d\setminus\{0\}$ such that $W(z_1)\neq W(z_2)$.
 Let $U$ satisfy
\cref{hyp:external-form,hyp:local-propagation}, and assume that there is
a point $y_0\in\IR^d$ such that
$y_0+z_1\in\Omega_U$ and $y_0+z_2\in\Omega_U$.
Then there are normalized orthogonal functions
$f_1,f_2\in C_c^\infty(\Omega_U)$ and a number $\theta>0$  such that
\begin{align}
    \liminf_{n\to\infty}
    \norm{\tau_{t_n}(A)-A}>0, \qquad \text{for }~ t_n= \frac{\theta}{n}, ~A = a^*(f_1)a(f_2).
    \label{eq:main-transfer-discontinuity}
\end{align}
Moreover, there exist
$h\in C_c^\infty(\Omega_U)$ and a subsequence
$(t_{n_k})$ such that $\liminf_{k\to\infty}
\norm{\tau_{t_{n_k}}(A)-A}>0$ for $A=n_h$ and $A =a^*(h)$.

\end{thm}
The failure of norm continuity implies a failure of invariance.
\begin{cor}[Non-invariance of the CAR algebras]
	\label{thm:main-noninvariance}
	Under the assumptions of \cref{thm:main-discontinuity}, neither
	$\Af^0_{\mathrm{CAR}}$ nor $\Af_{\mathrm{CAR}}$ is invariant under
	$\tau_t$. More precisely, for any $A$ chosen as in
	\cref{thm:main-discontinuity},
	\[
	\tau_t(A)\notin\Af_{\mathrm{CAR}}
	\qquad\text{for almost every }t\in\IR .
	\]
\end{cor}
At this point the special case of constant interaction potentials $W \equiv c$ should also be mentioned. Here the situation is very simple and does not require the machinery needed in the general case.
\begin{prop}[Constant interactions]
	\label{prop:constant-main}
	Suppose $W\equiv c\neq0$, and let $h$ be any self-adjoint
	one-particle Hamiltonian. 
	Then, for every normalized $f\in\fh$, there is a sequence $t_n\to0$
	such that
	\[
	\lim_{n\to\infty}
 \norm{\tau_{t_n}(a^*(f))-a^*(f)}=2.
	\]
	Therefore, the ordinary CAR algebra is not invariant with non-invariance occurring for almost every time $t\in \IR$. 
	
	On the other hand, the gauge-invariant CAR algebra is invariant, and the restriction of the dynamics to this algebra coincides with the free Bogoliubov dynamics, which is pointwise-norm continuous. More precisely, \[ \tau_t(A) = e^{\i t\dG(h)}Ae^{-\i t\dG(h)}, \qquad \text{ for all } A\in\Af^0_{\mathrm{CAR}}. \]
\end{prop}
We can summarize the previous results as a  dichotomy.  
\begin{rem}
	\label{rem:invariance-dichotomy}
	Assume $U=0$, let $\omega$ satisfy
	\cref{hyp:dispersion}, and let $W$ satisfy
	\cref{hyp:pair-potential}.  Suppose, in addition, that $W$ has a
	representative which is continuous on
	$\IR^d\setminus\{0\}$.  Then
	\begin{align*}
	\tau_t\bigl(\Af^0_{\mathrm{CAR}}\bigr)
	=
	\Af^0_{\mathrm{CAR}}
	\quad\text{for every }t\in\IR &\Longleftrightarrow W \text{ is constant almost everywhere}, \\
	\tau_t(\Af_{\mathrm{CAR}})
	=
	\Af_{\mathrm{CAR}}
	\quad\text{for every }t\in\IR &\Longleftrightarrow W=0 \text{ almost everywhere.} 
\end{align*}
	In fact,  if $W=c$ almost everywhere, then \Cref{prop:constant-main} implies invariance of $\Af^0_{\mathrm{CAR}}$.
	 If $W$ is not constant almost everywhere, continuity away
	from the origin yields two distinct nonzero points $z_1,z_2$ such
	that $W(z_1)\neq W(z_2)$.  Hence
	\cref{thm:main-noninvariance} shows that neither CAR algebra is
	invariant.  
 If $c=0$,
	the full CAR algebra is trivially invariant, whereas for $c\neq0$
	\cref{prop:constant-main} shows that the full CAR algebra is not
	invariant.
	
The continuity assumption in this remark is used only to ensure
that every potential which is not constant almost everywhere
has two nonzero continuity points at which its values differ, a requirement for \Cref{thm:main-discontinuity}.
\end{rem}

\section{Proofs}
\label{sec:proof}

The mechanism behind the proof is a relative-phase argument.  We first
construct a dense particle cluster in a  bounded region while its kinetic and
interaction energies remain of order at most $n^2$.    Two additional
one-particle modes are placed at different relative positions
$z_1,z_2$ with $W(z_1)\neq W(z_2)$.
On the time scale $t\sim n^{-1}$, their interactions with the same
$n$-particle cluster generate the two distinct phases
$e^{-2\i nW(z_1)t}$ and $e^{-2\i nW(z_2)t}$.

The resulting relative phase remains of order one as $n\to\infty$.
A quadratic leakage estimate shows that the cluster remains localized on
this time scale, while a localized antisymmetrization estimate avoids the factor $\sqrt n$ arising from the norm of the
full antisymmetrization map.  These bounds allow the relative phase to
be detected in a fixed CAR observable.

\subsection{Packing fermionic states and the leakage estimate}

The cluster construction is an elementary localized version of the standard
Slater determinant constructions with kinetic energy upper
bounds, see e.g.
\cite[eq. (3.46)]{LewinLiebSeiringer2023}. 

For technical reasons, we introduce the truncations
\[
W^{(R)}(x):=\min\{W(x),R\},
\qquad R>0,
\]
which are bounded by \cref{hyp:pair-potential}. We denote the corresponding quadratic forms and
Hamiltonians by $\Ec_n^{(R)}$ and $H_n^{(R)}$, respectively. Further, for two sets $A,B$ we write $A\Subset B$ if $\overline A$ is compact and $\overline A \subset B$. 
\begin{lem}[Packing Slater determinants]
	\label{lem:packed-states}
	Assume \cref{hyp:dispersion}.  Let $B_0\subset\IR^d$ be a nonempty bounded open
	set, and suppose that
	$W\in L^2(\IR^d)+L^\infty(\IR^d)$
	and $U\in L^1_{\mathrm{loc}}(\IR^d)$. 
	There exists $C>0$ such that, for every $n\in\NN$, one can find
	orthonormal functions
	$\varphi_{1,n},\ldots,\varphi_{n,n}\in C_c^\infty(B_0)$
	for which the normalized Slater determinant
	$\Psi_n=\varphi_{1,n}\wedge\cdots\wedge\varphi_{n,n}$
	satisfies
	\begin{equation*}
		\Ec^{(R)}_{n}[\Psi_n]\leq \Ec_n[\Psi_n]\leq C n^2 \quad \text{ for all } R > 0.
	\end{equation*}
\end{lem}
\begin{proof}
		Choose an open cube $Q\Subset B_0$ of side length $0 < L\leq1$, set
	$m_n=\left\lceil n^{1/d}\right\rceil$, 
	and divide $Q$ into $m_n^d$ mutually disjoint subcubes of side
	length $\ell_n:=L/m_n$.  In $n$ distinct subcubes place
	appropriately rescaled copies of a fixed normalized function
	$\varphi\in C_c^\infty((0,1)^d)$.  The resulting functions
	$\varphi_{1,n},\ldots,\varphi_{n,n}$ have disjoint supports and are
	therefore orthonormal. By scaling and
	\Cref{hyp:dispersion}, we get for all $j$,
	\[
	\langle\varphi_{j,n},\omega(\po)\varphi_{j,n}\rangle
	=\int_{\IR^d}\omega(p/\ell_n)\abs{\widehat\varphi(p)}^2\,\d p
	\leq
	C\ell_n^{-\nu}
	\int_{\IR^d}\langle p\rangle^\nu
	\abs{\widehat\varphi(p)}^2\,\d p
	\leq Cm_n^\nu.
	\]
	Hence, we get the required bound for the kinetic energy,
	\[
	\sum_{j=1}^n
	\langle\varphi_{j,n},\omega(\po)\varphi_{j,n}\rangle
	\leq C n m_n^\nu
	\leq Cn^{1+\nu/d}.
	\]
	
	Next, in order to estimate the interaction energy, we introduce the one-particle density matrix
	$
	\gamma_n(x,y)
	:=
	\sum_{j=1}^n
	\varphi_{j,n}(x)\overline{\varphi_{j,n}(y)},
	$
	and the corresponding one-particle density
	$\rho_n(x)= \gamma_n(x,x) =\sum_{j=1}^n|\varphi_{j,n}(x)|^2$, 
	 supported in $Q$.  Since the supports of the orbitals are
	disjoint, the same scaling yields
	\[
	\|\rho_n\|_\infty\leq C m_n^d\leq Cn,
	\qquad
	\int_{\IR^d}\rho_n(x)\,\d x=n.
	\]
	
With the usual decomposition into direct and exchange terms used in Hartree--Fock theory \cite{benedikter2020optimal}, we obtain
	\begin{align*}
		\Abs{\sum_{\substack{1\leq i,j\leq n\\i\neq j}}
		\int_{\IR^{dn}}W(x_i-x_j)\abs{\Psi_n(X)}^2\,\d X}
		&=  \Abs{	\iint_{Q\times Q} W(x-y)(        \rho_n(x)\rho_n(y)-|\gamma_n(x,y)|^2
			)  \,\d x\,\d y }
	\\
		&\leq
		\iint_{Q\times Q}
		|W(x-y)|\rho_n(x)\rho_n(y)\,\d x\,\d y \\
		&\leq
		\|\rho_n\|_\infty
		 \underbrace{ 		\left( \sup_{x\in Q}\int_Q|W(x-y)|\,\d y
		\right)}_{< \infty}
		\int_Q\rho_n(x)\,\d x \\
		&\leq Cn^2,
	\end{align*}
	where the supremum is finite because $W\in L^2+L^\infty$ and $Q-Q$ is bounded. Likewise,
	\begin{align*}
		\left|
		\left\langle
		\Psi_n,\sum_{j=1}^n U(x_j)\Psi_n
		\right\rangle
		\right|
		&=
		\left|\int_Q U(x)\rho_n(x)\,\d x\right| \leq
		\|\rho_n\|_\infty\|U\|_{L^1(Q)}
		\leq Cn.
	\end{align*}
	Combining the kinetic, interaction, and external-potential estimates, using
	$n^{1+\nu/d}\leq n^2$ because $\nu\leq d$,
	proves the statement. Note that $W^{(R)}\leq W$ implies
	$
	\Ec_n^{(R)}[\Psi_n] \leq\Ec_n[\Psi_n]$.
\end{proof}

As a next step, we study the stability of this cluster for time $t > 0$ through the following particle number leakage estimate. The proof uses some elementary ideas from the ASTLO technique proving propagation bounds, namely the control of the leaking observable by its derivative, a differential Gronwall-like inequality, and a smooth approximation of the indicator function, see e.g.
\cite{arbunich2021maximal} and
\cite[proof of Prop.~3.1]
{hinrichs2024lieb}.
\begin{lem}[Leakage estimate]
	\label{lem:quadratic-leakage}
	Assume \cref{hyp:dispersion,hyp:pair-potential,hyp:external-form}. Let
	$B_0\Subset B\subset\IR^d$, with $B$ open, and let
	$A\subset\IR^d$ be measurable with
	$A\cap\overline B=\varnothing$. If $\Psi_n \in \Def(\Ec_n) $ are normalized satisfying
	\[
	\supp\Psi_n\subset B_0^n,
	\qquad
	\Ec_n[\Psi_n]\leq C_0n^2,
	\]
	then
	\[
			\left\langle
		\Nf_A
		\right\rangle_{n,s} \leq \left\langle
		\Nf_{B^c}
		\right\rangle_{n,s}
		\leq Cn^2s^2,
	\qquad s\in\IR,
	\]
	where $\langle\,\cdot\,\rangle_{n,s}$ denotes expectation in
	$e^{-\i sH_n}\Psi_n$, and $C$ is independent of $n$ and $s$. The same estimate holds when the expectations are taken in
	$e^{-\i sH_n^{(R)}}\Psi_n$, with $C$ independent also
	of $R>0$.
\end{lem}

\begin{proof}
Choose $\chi,\rho\in C^\infty(\IR^d;[0,1])$ such that
$\chi-1,\rho\in C_c^\infty(B)$,
$\chi$ vanishes on a neighborhood of $\overline{B_0}$, and
$
\chi^2+\rho^2=1.
$
In particular, $\chi^2=1$ on $B^c$.
We control the leakage through
$\langle\dG(\chi^2)\rangle_{n,s}$.

Before taking the time derivative, we first have to show that $\dG(\chi^2)$ leaves the form domain $\Def(\Ec_n)$ invariant. 
Since $\chi-1,\rho\in C_c^\infty(\IR^d)$,
\eqref{eq:dispersion est} applies to both commutators
$[\omega(\po),\chi]$ and $[\omega(\po),\rho]$.
For $u\in C_c^\infty(\IR^d)$, using
$\chi^2+\rho^2=1$ and Cauchy--Schwarz, we obtain
\begin{align*}
	\norm{\omega(\po)^{1/2}\chi u}^2+\norm{\omega(\po)^{1/2}\rho u}^2
	&=
	\norm{\omega(\po)^{1/2}u}^2
	+\operatorname{Re}\langle\chi u,[\omega(\po),\chi]u\rangle
	+\operatorname{Re}\langle\rho u,[\omega(\po),\rho]u\rangle\\
	&\leq
	\norm{\omega(\po)^{1/2}u}^2
	+C\norm{u}\norm{(\omega(\po)+1)^{1/2}u}\\
	&\leq C'\norm{(\omega(\po)+1)^{1/2}u}^2.
\end{align*}
Discarding the nonnegative term
$\norm{\omega(\po)^{1/2}\rho u}^2$ and using
$\norm{\chi u}\leq\norm{u}$, we conclude that
\begin{align}
	\label{eq:cutoff bound}
\norm{(\omega(\po)+1)^{1/2}\chi u}
\leq C''\norm{(\omega(\po)+1)^{1/2}u}.
\end{align}
Under \cref{hyp:dispersion}, $C_c^\infty(\IR^d)$ is a form
core for $\omega(\po)$, so closedness of $\omega(\po)^{1/2}$ extends this estimate
to every $u\in\Def(\omega(\po)^{1/2})$.
Thus multiplication by $\chi$, and hence by $\chi^2$, is
bounded on the kinetic form domain.
Applying this in each coordinate shows that
$\dG(\chi^2)$ preserves the kinetic form domain.
As a bounded multiplication operator, it also preserves the
potential form domains. Therefore, $\dG(\chi^2)$ is
bounded on $\Def(\Ec_n)$ with its form norm.

The same density argument also gives, for all $u\in\Def(\omega(\po)^{1/2})$,
\begin{align}
	\label{eq:one particle comm}
	\i\Bigl(
	\langle \omega(\po)^{1/2}u,\omega(\po)^{1/2}\chi^2u\rangle
	-
	\langle \omega(\po)^{1/2}\chi^2u,\omega(\po)^{1/2}u\rangle
	\Bigr)
	=
	2\operatorname{Re}\langle\chi u,\i[\omega(\po),\chi]u\rangle.
\end{align}
Indeed, this identity follows by expansion for
$u\in C_c^\infty(\IR^d)$, and both sides are continuous
in the kinetic form norm by \eqref{eq:cutoff bound} and
\eqref{eq:dispersion est}.
Here $[\omega(\po),\chi]$ denotes its continuous extension
from $C_c^\infty(\IR^d)$ to
$\Def(\omega(\po)^{1/2})$.

	Write $\Psi_n(s):=e^{-\i sH_n}\Psi_n$, $E_n := \inf \sigma(H_n)$,  and 
	\[
	\Psi_{n,\ell}:=\chr_{[0,\ell]}(H_n - E_n)\Psi_n.
	\]
	Then $\Psi_{n,\ell}\in\Def(H_n)$, and
	\begin{align}
		\label{eq:Kn form conv}
	\norm{
		(H_n - E_n+1)^{1/2}e^{-\i sH_n}
		(\Psi_{n,\ell}-\Psi_n)
	}
	=
	\norm{(H_n - E_n+1)^{1/2}(\Psi_{n,\ell}-\Psi_n)}
	\longrightarrow0
\end{align}
	uniformly in $s$.
Set $\Psi_{n,\ell}(s):=e^{-\i sH_n}\Psi_{n,\ell} \in \Def(H_n) $ . Then
\begin{align}
	&\frac{\d}{\d s}
	 \cs{ \Psi_{n,\ell}(s),
	\dG(\chi^2)\Psi_{n,\ell}(s) } \nonumber \\
	&\quad=
	\i \left( \cs{ H_n 
	\Psi_{n,\ell}(s),
	\dG(\chi^2)\Psi_{n,\ell}(s) }
	- \cs{ 
	\Psi_{n,\ell}(s),
	\dG(\chi^2)  H_n  \Psi_{n,\ell}(s) }
	\right)  \nonumber \\
	&\quad= \i\Bigl( \Ec_n\bigl( \Psi_{n,\ell}(s), \dG(\chi^2)\Psi_{n,\ell}(s) \bigr) - \Ec_n\bigl( \dG(\chi^2)\Psi_{n,\ell}(s), \Psi_{n,\ell}(s) \bigr) \Bigr) \nonumber \\
	&\quad= 	\i \left( \cs{ \dG(\omega(\po))^{1/2} 
		\Psi_{n,\ell}(s),
		  \dG(\omega(\po))^{1/2}  \dG(\chi^2)\Psi_{n,\ell}(s) }
	- \textbf{h.c.}
\right) \nonumber \\
	&\quad=
	2\operatorname{Re}\sum_{j=1}^n
	\langle\chi_j\Psi_{n,\ell}(s),
	\i[\omega(\po),\chi]_j\Psi_{n,\ell}(s)\rangle, \label{eq:form derivative}
\end{align}
where the subscript $j$ denotes the action of the respective operator on the $j$-th component of the tensor product.  The potential contributions in the penultimate equality cancel because the potentials
and $\dG(\chi^2)$ are real multiplication operators.
The last equality follows from \eqref{eq:one particle comm}, applied in each coordinate, and terms acting in different coordinates commute.

We can now take the limit $\ell \to \infty$ in the first and last expressions in \eqref{eq:form derivative}  and replace $\Psi_{n,\ell}(s)$ by $\Psi_{n}(s)$: first integrate, then take the limit using \eqref{eq:Kn form conv}, \eqref{eq:dispersion est} and that full form norm convergence implies convergence in the
kinetic form norm, and finally differentiate again. 
Then Cauchy--Schwarz, followed by
\eqref{eq:dispersion est}, yields, for every $s$, the  differential inequality
\begin{align}
	\left|
	\frac{\d}{\d s}
	\left\langle\dG(\chi^2)\right\rangle_{n,s}
	\right|
	&=
	2 \Abs{ \Re \sum_{j=1}^n
	\langle\chi_j\Psi_n(s),\i[\omega(\po),\chi]_j\Psi_n(s)\rangle } \nonumber \\
	&\leq
	2
	\left\langle\dG(\chi^2)\right\rangle_{n,s}^{1/2}
	\left(
	\sum_{j=1}^n
	\norm{[\omega(\po),\chi]_j \Psi_n(s)}^2
	\right)^{1/2} \notag\\
	&\leq
	C_\chi
	\left\langle\dG(\chi^2)\right\rangle_{n,s}^{1/2}
	\left(
	\left\langle\dG(\omega(\po))\right\rangle_{n,s}+n
	\right)^{1/2}.
\label{eq:current-estimate}
\end{align}
	To control the kinetic-energy term, we use
	\Cref{hyp:external-form} and the lower bound on $W$,
	which imply the quadratic-form inequality
	\[
	H_n
	\geq
	(1-a)\dG(\omega(\po))-bn-C_Wn(n-1).
	\]
	Therefore, due to energy conservation and the assumed initial-energy bound,
	\begin{align}
	\sup_{s\in\IR}
	\left\langle\dG(\omega(\po))\right\rangle_{n,s}
	\leq
	\frac{C_0+b+C_W}{1-a}\,n^2
	\leq Cn^2.
	\label{eq:uniform-kinetic}
	\end{align}
	
	Because $\chi$ vanishes near the support of $\Psi_n$,
	$\left\langle\dG(\chi^2)\right\rangle_{n,0}=0.$
	Dividing \eqref{eq:current-estimate} by
	$(
	\left\langle\dG(\chi^2)\right\rangle_{n,s}
	+\delta
	)^{1/2}$, 
	integrating, and finally letting $\delta\downarrow0$, we obtain
	\[
	\left\langle\dG(\chi^2)\right\rangle_{n,s}^{1/2}
	\leq
	C_\chi\abs{s}
	\sup_{r\in\IR}
	\left(
	\left\langle\dG(\omega(\po))\right\rangle_{n,r}+n
	\right)^{1/2}
	\leq Cn\abs{s},
	\]
	where we used \eqref{eq:uniform-kinetic} in the last step.
	Finally, $A\subset B^c$ and $\chi^2=1$ on $B^c$, so
	\[
	\Nf_A
	\leq
	\Nf_{B^c}
	\leq\dG(\chi^2).
	\]
	Taking expectations proves the claimed leakage estimate.
	
	The proof applies equally to $H_n^{(R)}$, since
	\[
	W^{(R)}\geq-C_W,
	\qquad
	\Ec_n^{(R)}[\Psi_n]
	\leq\Ec_n[\Psi_n]\leq C_0n^2.
	\]
	The constants in the leakage estimate depend on the interaction
	only through $C_W$ and $C_0$, and are therefore independent
	of $R$.
\end{proof}

\subsection{Antisymmetrization and the cross-interaction}

We temporarily regard the added particle as distinguishable from the $n$-particle cluster. To this end, we work on $\fh\otimes\bigwedge\nolimits^n\fh$ and consider
\[
J_n:=\sqrt{n+1}\,P_-^{n+1},
\]
where $P_-^{n+1} \:
\fh\otimes\bigwedge\nolimits^n\fh
\longrightarrow
\bigwedge\nolimits^{n+1}\fh$ denotes the full antisymmetrization projection, so that
\begin{align}
J_n(f\otimes\Psi)=a^*(f)\Psi.
\label{eq:J-creation}
\end{align}
A naive estimate yields $\norm{J_n}\leq\sqrt{n+1}$, which adds $\sqrt n$
into every interaction estimate. However, this bound is only sharp if the distinguished particle can overlap with all other particles. 
 If it lies in $A$, then only particles already lying in $A$ matter. 
\begin{lem}[Localized antisymmetrization]
	\label{lem:localized-antisymmetrization}
	Let $A\subset\IR^d$ be measurable.  If
	$\Psi\in L^2(A)\otimes\bigwedge\nolimits^n\fh$, then
	\begin{align*}
	\norm{J_n\Psi}^2
	\leq
	\langle\Psi,(\Id\otimes(\Nf_A+\Id))\Psi\rangle.
	\end{align*}
	Further, if $A\cap B=\varnothing$, then $J_n$ is isometric on
	$
	L^2(A)\otimes\bigwedge\nolimits^nL^2(B).
	$
\end{lem}
\begin{proof}
	We use the canonical unitary identification
	\[
	\bigwedge^n L^2(\IR^d)
	\cong
	\bigoplus_{k=0}^n
	\bigwedge^kL^2(A)\otimes\bigwedge^{n-k}L^2(A^c),
	\]
	and write $\Psi=\sum_{k=0}^n\Psi_k$, where $\Psi_k \in L^2(A) \otimes \bigwedge^k L^2(A) \otimes \bigwedge^{n-k} L^2(A^c) $ has
	exactly $k$ cluster particles in $A$. Then
	$J_n$ acts on the $k$-th component as
	\[
	\sqrt{k+1}\,P_{-,A}^{k+1}\otimes\Id_{\bigwedge^{n-k} L^2(A^c) },
	\]
	where $P_{-,A}^{k+1}$ is the orthogonal antisymmetrization
	projection on the $k+1$ factors in $L^2(A)$.
	Since the images have distinct particle numbers in $A$, they are orthogonal to each other, and hence
	\[
	\norm{J_n\Psi}^2
	=\sum_{k=0}^n\norm{J_n\Psi_k}^2
	\leq\sum_{k=0}^n(k+1)\norm{\Psi_k}^2
	=\langle\Psi,(\Id\otimes(\Nf_A+\Id))\Psi\rangle.
	\]
	If all cluster particles lie in $B$ disjoint from $A$,
	only $k=0$ occurs. Since $P_{-,A}^1=\Id$, the map
	is then isometric.
\end{proof}
The following lemma isolates the phase produced by a localized many-particle cluster. If the added particle is localized in $A$, the $n$ cluster particles are localized in the disjoint region $B$, and $W(x-y)$ is approximately equal to a constant $c$ on $A\times B$, then their cross-interaction is approximately the scalar $2nc$. Hence, for short times, the full evolution differs from the separated evolution of the added particle and the cluster essentially by the phase $e^{-2\i nct}$. The lemma quantifies the errors caused by the variation of $W$ and by particles leaking out of their localization regions.
 
Subsequently, we will write
 \[
 u_t=e^{-\i t\ho}u,
 \qquad
 \Psi_n(t)=e^{-\i tH_n}\Psi_n
 \]
 for the one-body and many-body time evolutions.
\begin{lem}[Phase generated by a localized cluster]
		\label{lem:localized-phase}
	Assume
	\cref{hyp:dispersion,hyp:pair-potential,hyp:external-form,hyp:local-propagation}.
	Let $A,B\subset\IR^d$ be open sets with disjoint closures and
	$A\Subset\Omega_U$. Suppose that, for some $c\in\IR$ and
	$\eps>0$,
	\begin{align}
	\abs{W(x-y)-c}\leq\eps
	\qquad\text{for a.e. }(x,y)\in A\times B.
	\label{eq:non constant interaction}
\end{align}
	Let $u\in C_c^\infty(A)$ be normalized, fix a nonempty bounded open
	set $B_0\Subset B$, and let $\Psi_n$ be the packed Slater
	determinants from \cref{lem:packed-states} supported in $B_0^n$.
	Then, for some $C>0$ independent of $n$,
	\[
	\norm{
		e^{-\i tH_{n+1}}a^*(u)\Psi_n
		-e^{-2\i nct}a^*(u_t)\Psi_n(t)}
	\leq
	2\eps nt+Cn^2t^2,
	\qquad 0\leq t\leq\frac1n.
	\]
	The same estimate holds with $H_{n+1}$ replaced by
	$H_{n+1}^{(R)}$ and $\Psi_n(t)$ replaced by
	$e^{-\i tH_n^{(R)}}\Psi_n$, with the same constant $C$,
	for every $R>\max\{0,c+\eps\}$.
\end{lem}

\begin{proof}
	On $\fh\otimes\bigwedge^n\fh$, set
	\[
	\widetilde{H}_n=\ho\otimes\Id+\Id\otimes H_n,
	\qquad
	\widetilde{W}_n=2\sum_{j=1}^nW(x-y_j).
	\]
	We first assume that $W$ is bounded, so that
	$\widetilde{W}_n$ is bounded and Duhamel's formula applies:
	\[
		e^{-\i t(\widetilde{H}_n+\widetilde{W}_n)}
		-e^{-\i t(\widetilde{H}_n+2nc)} =
		-\i\int_0^t
		e^{-\i(t-s)(\widetilde{H}_n+\widetilde{W}_n)}
		(\widetilde{W}_n-2nc)
		e^{-\i s(\widetilde{H}_n+2nc)}
		\,\d s.
	\]
	We then apply this identity to the vector 
	 $u\otimes\Psi_n$,  and let  $J_n$ act from the left.  
	Using \eqref{eq:J-creation} and the identities 
	\begin{align*}
J_n(\widetilde{H}_n+\widetilde{W}_n) & =H_{n+1}J_n, \quad  e^{-\i tH_{n+1}}a^*(u)\Psi_n = J_ne^{-\i t(\widetilde{H}_n+\widetilde{W}_n)}(u\otimes\Psi_n), \\ 	e^{-\i s(\widetilde{H}_n+2nc)} (u\otimes\Psi_n)	&= e^{-2\i ncs}(u_s\otimes\Psi_n(s)),
	\end{align*}
	we arrive at
	\[
		e^{-\i tH_{n+1}}a^*(u)\Psi_n
		-e^{-2\i nct}a^*(u_t)\Psi_n(t)
		=
		-\i\int_0^t
		e^{-\i(t-s)H_{n+1}}
		J_n(\widetilde{W}_n-2nc)
		e^{-2\i ncs}(u_s\otimes\Psi_n(s))
		\,\d s.
	\]
	Hence,
	\begin{align}
		\norm{
			e^{-\i tH_{n+1}}a^*(u)\Psi_n
			-e^{-2\i nct}a^*(u_t)\Psi_n(t)}
\leq
		\int_0^t
		\norm{
			J_n(\widetilde{W}_n-2nc)(u_s\otimes\Psi_n(s))
		}\,\d s.
		\label{eq:after duhamel}
\end{align}
	
	Choose $\chi\in C_c^\infty(A)$ with $0\leq \chi\leq1$ and $\chi=1$ near $\supp u$.  After some time $s$, $u_s$ need not remain supported anymore in $\supp u$, but the error can be controlled with \Cref{hyp:local-propagation} by
	\[
	\sup_{\abs{s}\leq1}\norm{\chi u_s}_{H^m}\leq C,
	\qquad
	\norm{(1-\chi)u_s}_{H^m}\leq C\abs{s}.
	\]
	Let
	\[ 
	P_n : =\chr_B^{\otimes n}\big|_{\bigwedge^n\fh}.
	\]
	Then we have $\Id-P_n \leq	\Nf_{B^c}$, and hence by 	\cref{lem:quadratic-leakage},
	\begin{align}
		\norm{(\Id-P_n)\Psi_n(s)}^2
		\leq
		\left\langle\Nf_{B^c}\right\rangle_{n,s} \leq  Cn^2s^2  
		\label{eq:qleak1}
	\end{align}
		Furthermore, $A\cap B=\varnothing$ implies $\Nf_AP_n	=P_n\Nf_A=0$, and therefore,
	\begin{align}
		\left\langle
		(\Id-P_n)\Psi_n(s),
		\Nf_A(\Id-P_n)\Psi_n(s)
		\right\rangle =  \left\langle\Psi_n(s),
		\Nf_A\Psi_n(s)
		\right\rangle
		\leq Cn^2s^2.
				\label{eq:qleak2}
\end{align}

	We now decompose the vector in the integrand of \eqref{eq:after duhamel} as
	\begin{align}
	u_s\otimes\Psi_n(s)	=\chi u_s\otimes P_n\Psi_n(s)+	\chi u_s\otimes(\Id-P_n)\Psi_n(s)+(1-\chi)u_s\otimes\Psi_n(s).
		\label{eq:three terms}
	\end{align}
We estimate the three terms in \eqref{eq:three terms} separately.  For the first one, notice that
the almost-constant interaction assumption \eqref{eq:non constant interaction} implies, almost everywhere on
$A\times B^n$,
\[
\abs{\widetilde{W}_n-2nc}
= 2\left| \sum_{j=1}^n\bigl(W(x-y_j)-c\bigr) \right| 
\leq 2\eps n.
\]
Moreover, $J_n$ is isometric on
$L^2(A)\otimes\bigwedge^nL^2(B)$ by
\cref{lem:localized-antisymmetrization}.  Hence, using
$\norm{\chi u_s}\leq\norm{u_s}=1$ and
$\norm{P_n\Psi_n(s)}\leq1$, the first one can be estimated by
\begin{align}
\norm{
		J_n(\widetilde{W}_n-2nc)
		\bigl(\chi u_s\otimes P_n\Psi_n(s)\bigr)
	} =
	\norm{
		(\widetilde{W}_n-2nc)
		\bigl(\chi u_s\otimes P_n\Psi_n(s)\bigr)
	}
	\leq 2\eps n.
	\label{eq:three est 1}
\end{align}

For the remaining two terms, first fix the number $m>d/2$ from
\cref{hyp:local-propagation}.  Sobolev embedding and
\cref{hyp:pair-potential} yield the translation-uniform multiplier
bound
\begin{align}
	\norm{W(\,\cdot-y)v}
	\leq
	\norm{W_2}\norm v_\infty
	+\norm{W_\infty}_\infty\norm v
	\leq K_m\norm v_{H^m},
	\qquad y\in\IR^d.
\label{eq:uniform-multiplier}
\end{align}
Thus, 
\eqref{eq:uniform-multiplier} and the triangle inequality imply
\begin{align}
	\norm{(\widetilde{W}_n-2nc)(v\otimes\Phi)}
	&=
	2\Norm{
		\sum_{j=1}^n
		\bigl(W(x-y_j)-c\bigr)(v\otimes\Phi)
	}
	\leq
	K_{m,c}n\norm{v}_{H^m}\norm{\Phi}
\label{eq:cross-before-antisym}
\end{align}
for every $v\in H^m(\IR^d)$ and
$\Phi\in\bigwedge^n\fh$.

For the second part in which at least one cluster particle has left $B$, set
\[
\widetilde\Phi_n(s)
=
(\Id-P_n)\Psi_n(s),
\qquad
\eta_s
=
(\widetilde{W}_n-2nc)
\bigl(\chi u_s\otimes\widetilde\Phi_n(s)\bigr).
\]
Since $\chi u_s$ is supported in $A$, the distinguished-particle
variable of $\eta_s$ is also supported in $A$.  Hence
\cref{lem:localized-antisymmetrization} gives
\[
\norm{J_n\eta_s}^2
\leq
\norm{\eta_s}^2
+
\left\langle
\eta_s,
\bigl(\Id\otimes\Nf_A\bigr)\eta_s
\right\rangle.
\]
Since $\widetilde{W}_n$ and $\Nf_A$ are multiplication operators, they
commute.  Applying \eqref{eq:cross-before-antisym} first to
$\widetilde\Phi_n(s)$ and then to
$\Nf_A^{1/2}\widetilde\Phi_n(s)$, we obtain
\begin{align}
	\norm{J_n\eta_s}^2
	\leq
	Cn^2\norm{\chi u_s}_{H^m}^2
	\Bigl(
	\norm{\widetilde\Phi_n(s)}^2
	+
	\left\langle
	\widetilde\Phi_n(s),
	\Nf_A\widetilde\Phi_n(s)
	\right\rangle
	\Bigr) 
\leq Cn^4 s^2,
\label{eq:three est 2}
\end{align}
where we used \Cref{hyp:local-propagation} to show that $\norm{\chi u_s}_{H^m}^2$ is uniformly bounded in $\abs s \leq 1$, and  the cluster-leakage estimates \eqref{eq:qleak1} and  \eqref{eq:qleak2}.

Finally, in the third term in \eqref{eq:three terms}, $(1-\chi)u_s$ need not be supported in $A$, so we use the trivial estimate
$\norm{J_n}\leq\sqrt{n+1}$ together with
\eqref{eq:cross-before-antisym}.  Since $(1-\chi)u=0$, \Cref{hyp:local-propagation} implies
\[
\norm{(1-\chi)u_s}_{H^m}
\leq C\abs{s},
\qquad \abs{s}\leq1.
\]
Moreover, $\norm{\Psi_n(s)}=1$.  Consequently,
\begin{align}
	\norm{
		J_n(\widetilde{W}_n-2nc)
		\bigl((1-\chi)u_s\otimes\Psi_n(s)\bigr)
	}
\leq
	Cn\sqrt{n+1}\,
	\norm{(1-\chi)u_s}_{H^m}
	\leq Cn^{3/2}\abs{s},
	\label{eq:three est 3}
\end{align}

	Collecting the three estimates \eqref{eq:three est 1}, \eqref{eq:three est 2} and \eqref{eq:three est 3}, and integrating them in \eqref{eq:after duhamel} gives
	\[
	2\eps nt+Cn^2t^2+Cn^{3/2}t^2
	\leq
	2\eps nt+Cn^2t^2,
	\]
	which proves the assertion for bounded $W$.
	
	For general $W$, apply this argument first to $W^{(R)}$.
	If $R>\max\{0,c+\eps\}$, then $W^{(R)}=W$ on the
	differences $x-y$ occurring in
	\eqref{eq:non constant interaction}, so that condition holds
	with the same $c$ and $\eps$.
	Moreover, $\abs{W^{(R)}}\leq\abs W$, and hence
	\[
	\norm{W^{(R)}(\,\cdot-y)v}
	\leq\norm{W(\,\cdot-y)v}
	\leq K_m\norm{v}_{H^m}
	\]
	with the same constant as in \eqref{eq:uniform-multiplier}.
	Together with the uniform statements in
	\cref{lem:packed-states,lem:quadratic-leakage}, this shows
	that the preceding argument gives the asserted truncated
	estimate with $C$ independent of $n$ and $R$.
	
	For each fixed particle number $n$, the closed forms
	$\Ec_n^{(R)}$ increase to $\Ec_n$ and have a
	common lower bound. Monotone convergence of quadratic forms
	therefore implies strong resolvent convergence
	$H_n^{(R)}\to H_n$, and thus,
	$
	e^{-\i tH_n^{(R)}}\Phi
	\to e^{-\i tH_n}\Phi
	$
	for every $\Phi\in\bigwedge^n\fh$ and every $t\in\IR$.
	Taking $R\to\infty$ in the truncated estimate, with $n$
	and $t$ fixed, and using boundedness of $a^*(u_t)$,
	proves the general statement for $W$.
\end{proof}
\subsection{Proof of the main discontinuity theorem}
\label{sec:main result proof}
\begin{proof}[Proof of \cref{thm:main-discontinuity}]
	After interchanging $z_1$ and $z_2$, if necessary, set
	$\delta := W(z_1)-W(z_2)>0$.
	Fix $0<\eps<\delta/2$.  Choose
	$r>0$ so small that
	\[
	B=B_r(y_0),
	\qquad
	A_j=B_r(y_0+z_j),
	\quad j=1,2,
	\]
	have pairwise disjoint closures,
	$A_1\cup A_2\Subset\Omega_U$, and
	\[
	\abs{W(x-y)-W(z_j)}\leq\eps,
	\qquad
	x\in A_j,\quad y\in B,\quad j=1,2,
	\]
	where the latter condition can be imposed due to continuity of $W$ at each $z_j$ and since $x-y\in B_{2r}(z_j)$.
	
	Choose normalized functions
	\[
	f_j\in C_c^\infty(B_{r/2}(y_0+z_j)),
	\qquad j=1,2.
	\]
	Their supports are disjoint, so $f_1\perp f_2$.  Let $\Psi_n$
	be the packed Slater determinants supported in
	$B_{r/2}(y_0)^{n}$.  Since the cluster orbitals are orthogonal to
	$f_1$ and $f_2$, the CAR imply
	$\norm{a^*(f_1)\Psi_n} = \norm{a^*(f_2)\Psi_n} = 1$ and $\langle a^*(f_1)\Psi_n,a^*(f_1)a(f_2)a^*(f_2)\Psi_n\rangle=1$.
The strategy is to lower bound the operator norm via the matrix element between the states $a^*(f_1)\Psi_n$ and $a^*(f_2) \Psi_n$, i.e.,
		\begin{align}
		\norm{
			\tau_t\bigl(a^*(f_1)a(f_2)\bigr)
			-a^*(f_1)a(f_2)
		} &\geq  \Abs{ \cs{	a^*(f_1)\Psi_n,
				(\tau_t\bigl(a^*(f_1)a(f_2)\bigr)
				-a^*(f_1)a(f_2) )	a^*(f_2)\Psi_n}} \nonumber \\
			&=   \Abs{ \cs{	a^*(f_1)\Psi_n,
					\tau_t\bigl(a^*(f_1)a(f_2)\bigr)	a^*(f_2)\Psi_n} - 1} 			\label{eq:first lower bound}.
	\end{align}
	The Heisenberg time-evolution of the first term on the right-hand side of \eqref{eq:first lower bound} can be moved to the states on the left and right. Subsequently, we approximate the fully time-evolved states $e^{-\i tH_{n+1}}a^*(f_j)\Psi_n$  by the separately evolved states $e^{-2\i nW(z_j)t}a^*(f_{j,t})\Psi_n(t)$:
		\begin{align}
		&\left\langle
		a^*(f_1) \Psi_n,
		\tau_t\bigl(a^*(f_1)a(f_2)\bigr)
		a^*(f_2)\Psi_n
		\right\rangle \nonumber
		\\ &\qquad  =  	\left\langle
		e^{-\i tH_{n+1}} 	a^*(f_1)\Psi_n,
		a^*(f_1)a(f_2)
		e^{-\i tH_{n+1}} a^*(f_2)\Psi_n
		\right\rangle  \nonumber \\
		&\qquad = 	e^{-2\i n(W(z_2) - W(z_1))t} \left\langle
		a^*(f_{1,t})\Psi_n(t),
		a^*(f_1)a(f_2)a^*(f_{2,t})\Psi_n(t)
		\right\rangle + R^{(1)}_n(t),  \label{eq: Wz2 Wz1 term}
	\end{align}
	where
	$
	f_{j,t}=e^{-\i t\ho}f_j
	$, $0\leq t\leq1/n$, 
	and where the error for the two replacements can be controlled with \cref{lem:localized-phase} as
	\begin{align*}
		\abs{ R^{(1)}_n(t) } \leq 4\eps nt+Cn^2t^2.
	\end{align*}
Note that the inner product in \eqref{eq: Wz2 Wz1 term} equals one for $t=0$, and we can control the error for positive time. \Cref{hyp:local-propagation} yields
	\[
	\norm{f_{1,t}-f_1}+\norm{f_{2,t}-f_2}
	\leq C\abs t.
	\]
	Then we use the CAR, the fact that $n_{f_1}$ and $n_{f_2}$ commute, and \cref{lem:quadratic-leakage} to obtain
	\begin{align*}
		1 - \left\langle
		a(f_1)a^*(f_1)a(f_2)a^*(f_2)
		\right\rangle_{n,t} &=
		\left\langle
		\left( 1 - 
		(1-n_{f_1})(1-n_{f_2}) \right)
		\right\rangle_{n,t} \leq
		\left\langle
		 \left( n_{f_1} + n_{f_2}  \right)
		\right\rangle_{n,t} \\
		&\leq \left\langle
		\left( \dG(\chr_{A_1} + \chr_{A_2})  \right)
		\right\rangle_{n,t} \leq C	n^2 t^2. 
\end{align*}
	Combining the last two estimates yields
\begin{align}
	\left\langle
		a^*(f_{1,t})\Psi_n(t),
		a^*(f_1)a(f_2)a^*(f_{2,t})\Psi_n(t)
		\right\rangle
		= 1 + R^{(2)}_n(t), \qquad \abs{R^{(2)}_n(t)} \leq C\bigl(\abs t+n^2t^2\bigr).
\label{eq:separated-matrix-one}
\end{align}
Inserting \eqref{eq:separated-matrix-one} into  \eqref{eq: Wz2 Wz1 term}, and using $W(z_1)-W(z_2)=\delta$ in the exponent,  we arrive at 
	\[
		\left\langle
		a^*(f_1)\Psi_n,
		\tau_t\bigl(a^*(f_1)a(f_2)\bigr)
		a^*(f_2)\Psi_n
		\right\rangle
 =
		e^{2\i n\delta t}   + R^{(1)}_n(t) +  e^{2\i n\delta t} R_n^{(2)}(t).
	\]
Inserting this into \eqref{eq:first lower bound}, and using the error estimates, we obtain
	\begin{align*}
		\norm{
			\tau_t\bigl(a^*(f_1)a(f_2)\bigr)
			-a^*(f_1)a(f_2)
		} \geq 
		\abs{e^{2\i n\delta t}-1}	-4\eps nt	- C\bigl(n^2t^2+\abs t\bigr).
\end{align*}
		Finally, set $t_n=\theta/n$, where $0<\theta\leq1$.  Then
\[
	\liminf_{n\to\infty}
	\norm{
		\tau_{t_n}\bigl(a^*(f_1)a(f_2)\bigr)
		-a^*(f_1)a(f_2)
	}
	\geq
	2\abs{\sin(\delta\theta)}-4\eps\theta-C\theta^2
	\geq
	\theta\Bigl(2\delta-4\eps-C\theta-\tfrac{\delta^3}{3}\theta^2\Bigr).
\]
Since $\eps<\delta/2$, the right-hand side is strictly positive
for all sufficiently small $\theta>0$, which shows \eqref{eq:main-transfer-discontinuity}.

Finally, the last statement follows from polarization. Writing
	\[
	a^*(f_1)a(f_2)	= \frac12\sum_{\ell=0}^3 \i^\ell n_{h_\ell}, \qquad \text{ with } h_\ell :=\frac{f_1+\i^\ell f_2}{\sqrt2},
	\]
	we have
	\[
		\norm{
			\tau_t\bigl(a^*(f_1)a(f_2)\bigr)
			-a^*(f_1)a(f_2)
		} \leq
		\frac12 \sum_{\ell=0}^3
		\norm{\tau_t(n_{h_\ell})-n_{h_\ell}} \leq \sum_{\ell=0}^3
		\norm{\tau_t(a^*(h_\ell))-a^*(h_\ell)}.
	\]
	Along the sequence $(t_n)$, at least one of the four terms in the middle (and therefore also on the right) is
	bounded away from zero on a subsequence.  
\end{proof}

\subsection{Non-invariance and constant interactions}
\label{sec:noninvariance}

The non-invariance follows from the discontinuity result, the separability of
the CAR algebra over a separable Hilbert space, and a classical automatic-continuity argument. The following lemma and its proof are a variant of \cite[Theorem 10.2.3]{hille_phillips}. It is restricted to isometric groups, but requires only a single orbit to be separably valued on a general set of positive Lebesgue measure.

\begin{lem}[Automatic continuity]
	\label{lem:automatic-continuity}
	Let $(U_t)_{t\in\IR}$ be a strongly continuous unitary group on a Hilbert space $\cH$, and set $\tau_t(A):=U_tAU_t^*$ for
	$A\in\cB(\cH)$. Fix $A\in\cB(\cH)$ and let
	$\mathfrak X\subseteq\cB(\cH)$ be a norm closed, norm separable
	subspace. Then
	\[
	E:=\{s\in\IR:\tau_s(A)\in\mathfrak X\}
	\]
	 is a Borel set. If $E$ has positive Lebesgue measure,
	then $\IR \rightarrow \cB(\cH),~ t\mapsto\tau_t(A)$ is norm continuous.
\end{lem}
\begin{proof}
	For each $C\in\cB(\cH)$, the function $s\mapsto\norm{\tau_s(A)-C} = \sup_{\nn{\xi} = 1}\norm{(\tau_s(A)-C)\xi} $ is a
	supremum of continuous functions and hence Borel measurable. 
	If $(C_k)_{k\in\IN}$ is a dense sequence in $\mathfrak X$, then we can write
	\[E=\{s:\inf_k\norm{\tau_s(A)-C_k}=0\},\] 
	so $E$ is a Borel set.
	
	Now assume $\abs E>0$, where $\abs{\,\cdot\,}$ denotes the Lebesgue
	measure. We further assume that $E$ is a bounded set, otherwise shrink it to a smaller set with positive measure. 
	 Consider the function
	\begin{align*}
		g \: \IR \rightarrow \mathfrak X, \quad g(s) := \chr_{E}(s)\tau_{s}(A).
	\end{align*}
	By the first step, $s\mapsto\norm{g(s)-C}$ is Borel measurable for every $C\in\mathfrak X$. Since  $\mathfrak X$ is separable, this implies that $g$ is also Borel measurable: every open subset of $\mathfrak X$ is a countable union of	balls $B_r(C)$, $r >0$, $C \in \mathfrak X$, and $g^{-1}(B_r(C)) = \{s \in \IR : \nn{g(s)-C} < r\}$. Then $g$ is weakly measurable and separably valued, and hence Bochner measurable by Pettis's theorem \cite[Theorem 1.1]{pettis_measurable}.
	As  $\norm{g(s)}\leq\nn A$ for all $s$ and $E$ is bounded, we get
	$g\in L^1(\IR;\mathfrak X)$. 
	
	Note that by the definition of $g$, we have for all $h\in \IR$ and $s\in E\cap(E-h)$,
	\[
	\norm{g(s+h)-g(s)}
	=\norm{\tau_{s+h}(A)-\tau_s(A)}
	=\norm{\tau_h(A)-A}.
	\]
	Hence,
	\[
	\abs{E\cap(E-h)}\,\norm{\tau_h(A)-A}
	\leq\int_\IR\norm{g(s+h)-g(s)}\,\d s \overset{h\to 0}\longrightarrow 0
	\]
	where the last step follows from continuity of translations in the $L^1$-Bochner space, see e.g. \cite[Theorem 3.8.3]{hille_phillips}. The first factor is positive, as $\abs{E\cap(E-h)} \to \abs E > 0$ for $h \to 0$. Thus, $\norm{\tau_h(A)-A}\to0$, and continuity at
	every $t\in\IR$ follows from
	$\norm{\tau_{t+h}(A)-\tau_t(A)}=\norm{\tau_h(A)-A}$.
\end{proof}
The non-invariance result now follows immediately as the contrapositive of this lemma. 
\begin{proof}[Proof of \cref{thm:main-noninvariance}]
	Since $L^2(\IR^d)$ is separable, the CAR algebra
	$\Af_{\mathrm{CAR}}$ is separable as well, cf.\ \cite[Theorem 5.2.5]{BR2}. If the Borel set
	$\{t\in\IR:\tau_t(A)\in\Af_{\mathrm{CAR}}\}$ had positive Lebesgue
	measure, \cref{lem:automatic-continuity} with
	$\mathfrak X=\Af_{\mathrm{CAR}}$ would imply that
	$t\mapsto\tau_t(A)$ is norm continuous, contradicting
	\cref{thm:main-discontinuity}. 	Since $a^*(f_1)a(f_2)\in\Af^0_{\mathrm{CAR}}\subset\Af_{\mathrm{CAR}}$, neither algebra
	is invariant.
\end{proof}

\begin{proof}[Proof of \cref{prop:constant-main}]
If $W \equiv c$, we can write
	\[
	\Hb=\dG(h)+c\Nf(\Nf-1).
	\]
	The two operators commute, and as $\Nf a^*(f) = a^*(f)(\Nf+1)$, we have
	\begin{align}
	\tau_t(a^*(f))
	=  	e^{\i t c\Nf(\Nf-1)} a^*(e^{\i th}f) e^{-\i t c\Nf(\Nf-1)}   =
	a^*(e^{\i th}f)e^{2\i ct\Nf}.
\label{eq:constant-pull-through}
	\end{align}
	Set
	$t_n=\frac{\pi}{2\abs c\,n}$. 
	For every $n$, choose a normalized $n$-particle Slater determinant
	$\Psi_n$ whose orbitals are orthogonal to both $f$ and
	$e^{\i t_nh}f$.  This is possible because their common orthogonal
	complement in $L^2(\IR^d)$ is infinite-dimensional.  Since
	$e^{2\i ct_nn}=-1$, \eqref{eq:constant-pull-through} yields
	\[
		\norm{\bigl(\tau_{t_n}(a^*(f))-a^*(f)\bigr)\Psi_n}
	=
		\norm{a^*(e^{\i t_nh}f+f)\Psi_n}
		=
		\norm{e^{\i t_nh}f+f}
		\to 2.
	\]
	Therefore non-invariance of the ordinary CAR algebra again follows from \cref{lem:automatic-continuity}.
	
	Finally, the last statement follows from the fact that
	$c\Nf(\Nf-1)$ commutes with any $A \in  \Af^0_{\mathrm{CAR}} $.
\end{proof}

\subsection{Examples of external potentials}
\label{sec:potential-examples}

Here we give the missing proofs of the statements in \Cref{cor:external-examples}. 

\begin{proof}[Proof of \ref{it:external simple}]
Since $U\in W^{k,\infty}(\IR^d;\IR)$, it is bounded and hence \cref{hyp:external-form} holds with $a=0$.
	
	Note that the free group $e^{\i t \omega(\po)}$ is isometric on $H^k$ and the Leibniz rule implies
	\[
	\norm{Uv}_{H^k}\leq C_k\norm U_{W^{k,\infty}}\norm v_{H^k}.
	\]
	By Duhamel's formula, for $t\geq0$,
	\[
	e^{\i t\ho}v =e^{\i t \omega(\po)}v
	+\i\int_0^t
	e^{\i(t-s)\omega(\po)}Ue^{\i s\ho}v\,\d s.
	\]
Thus,
	\[
	\norm{e^{\i t\ho}v}_{H^k}	\leq \norm v_{H^k} +C_k\norm U_{W^{k,\infty}}\int_0^t \norm{e^{\i s\ho}v}_{H^k}\,\d s,
	\]
	and Gronwall's inequality, applied also to negative times, yields
	\begin{align}
		\label{eq:free group norm est}
	\norm{e^{\i t\ho}v}_{H^k}
	\leq
	e^{C_k\norm U_{W^{k,\infty}}\abs t}\norm v_{H^k}.
	\end{align}
	
	Next, let $u\in C_c^\infty(\IR^d)$.  By \Cref{hyp:dispersion}, $\omega(\po) u \in H^k$, and by assumption $U u \in H^k$, so $\ho u\in H^k$.  Using
	\[
	e^{\i t\ho}u-u
	=
	\i\int_0^t e^{\i s\ho}\ho u\,\d s
	\]
	and \eqref{eq:free group norm est}, we obtain, for $\abs t\leq1$,
	\[
	\norm{e^{\i t\ho}u-u}_{H^k}\leq e^{C_k\norm U_{W^{k,\infty}}}\abs t\,\norm{\ho u}_{H^k}.
	\]
	Therefore \cref{hyp:local-propagation} holds with
	$m=k>d/2$ and $\Omega_U=\IR^d$.
\end{proof}
\begin{proof}[Proof of \ref{it:external coulomb}]
 By 
Cauchy--Schwarz and Hardy's inequality, we have for every $R\in\IR^3$ and $\varepsilon > 0$,
\[
    \int_{\IR^3}\frac{\abs{v(x)}^2}{\abs{x-R}}\,\d x
    \leq
    \norm{\abs{x-R}^{-1}v}\norm v
    \leq
    2\norm{\nabla v}\norm v
    \leq
    \varepsilon\norm{\nabla v}^2
    +\varepsilon^{-1}\norm v^2.
\]
Since there are only finitely many centres,
\cref{hyp:external-form} follows by taking $\varepsilon$ sufficiently
small.  The operator version of Hardy's
inequality also shows that the Coulomb terms are infinitesimally
operator-bounded with respect to $-\Delta$.  Hence
$\Def(\ho)=H^2(\IR^3)$
with equivalent graph and $H^2$-norms.
Set
\[
    \Omega_U=\IR^3\setminus\{R_1,\ldots,R_M\}.
\]
If $u\in C_c^\infty(\Omega_U)$, then $u\in \Def(\ho^2)$.
Invariance of the graph domain and the spectral theorem yield
\[
    \norm{e^{\i t\ho}u-u}_{H^2}
    \leq
    C\norm{e^{\i t\ho}u-u}_{\Def(\ho)}
    \leq
    C\abs t\bigl(\norm{\ho u}+\norm{\ho^2u}\bigr),
    \qquad \abs t\leq1.
\]
Thus \cref{hyp:local-propagation} holds with $m=2$.  
\end{proof}
Finally, notice that the additional condition
\[
y_0+z_j\notin\{R_1,\ldots,R_M\},
\qquad j=1,2,
\]
in \Cref{thm:main-discontinuity}
excludes only finitely many choices of $y_0$, so an admissible
translation always exists.

\subsection*{Acknowledgment}

OpenAI's GPT-5.6 Sol and GPT-6 Astra, and Anthropic's Claude Opus 5.5 were used to draft proofs, to search for related literature and to proofread the manuscript. All content was revised and verified by the author, who takes full responsibility for it.

\printbibliography

\end{document}